\documentclass[letterpaper,11pt]{article}

\usepackage[runin]{abstract}
\usepackage{geometry}
\usepackage{titlesec}
\usepackage{graphicx}
\usepackage{amsmath}
\usepackage{amsthm}
\usepackage{amsfonts}
\usepackage{amssymb}
\usepackage{empheq}
\usepackage{algorithmic}
\usepackage{algorithm}
\usepackage[authoryear]{natbib}
\usepackage{url}

\usepackage[font=sf]{caption}

\titleformat*{\section}{\Large\bfseries}
\titleformat*{\subsection}{\large\sc}
\titleformat*{\subsubsection}{\itshape}

\begin{document}

\title{{\bf On the computational complexity of evolution}}

\author{{\large{ Ioannis Avramopoulos\footnote{The author is with the National Technical University of Athens. His email is \texttt{iavramop@central.ntua.gr.}}}}
}

\maketitle

\thispagestyle{empty} 

\newtheorem{definition}{Definition}
\newtheorem{proposition}{Proposition}
\newtheorem{theorem}{Theorem}
\newtheorem*{theorem*}{Theorem}
\newtheorem{corollary}{Corollary}
\newtheorem{lemma}{Lemma}
\newtheorem{axiom}{Axiom}
\newtheorem{thesis}{Thesis}

\vspace*{-0.2truecm}

\begin{abstract}
It is well-known that the problem of recognizing an ESS in a symmetric bimatrix game is {\bf coNP}-complete. In this paper, we show that recognizing an ESS even in doubly symmetric bimatrix games is also {\bf coNP}-complete. Our result further implies that recognizing asymptotically stable equilibria of the replicator dynamic in this class of games is also a {\bf coNP}-complete problem. 
\end{abstract}

\section{Introduction}
\label{introduction}

The theory of evolution, as that was formulated by Darwin, has had a profound impact in the sciences, not only in biology but also in the social sciences and economics. From a mathematical perspective, evolutionary theory relies on the framework of noncooperative games. From such an analytical perspective, evolutionary theory can be approached either from the perspective of notions of {\em evolutionary stability} \citep{TheLogicOfAnimalConflict, Evolution} or that of {\em evolutionary dynamics} (such as the {\em replicator dynamic} of \cite{TaylorJonker}). 

In this paper, we examine theories of evolution from a computational complexity perspective. \cite{Etessami} and \cite{Nisan-ESS} have cast doubt that the notion of evolutionary stability can serve as a general foundation of evolutionary theories unless {\bf P = NP}, by showing that the problem of recognizing an {\em evolutionarily stable strategy (ESS)} is {\bf coNP}-complete. This result applies to the general class of symmetric bimatrix games. 

Restricting the problem to the class of doubly symmetric (coordination) games (that is, symmetric bimatrix games where the payoff matrix is symmetric), one would perhaps expect that such a detection problem becomes computationally  tractable as coordination games were recently strongly tied to Darwin's theory of evolution: \cite{MWUA} view evolution as a coordination game solved by a multiplicative weights algorithm akin to the replicator dynamic (cf. \citep{AHK}) (see also \citep{MPP}).  Alas, we show that the problem remains {\bf coNP}-complete even in symmetric coordination games. Since a strategy in such a coordination game is an ESS if and only it is asymptotically stable under the replicator dynamic \citep{Hofbauer-Sigmund}, the problem of detecting asymptotically stable equilibria of the replicator dynamic is also {\bf coNP}-complete.

\section{Preliminaries}
\label{preliminaries}

The ESS is a refinement of the symmetric Nash equilibrium in symmetric bimatrix games. Let us, therefore, start off by introducing game-theoretic concepts in this setting.

\subsection{Nash equilibria}

A $2$-player (bimatrix) game in normal form is specified by a pair of $n \times m$ matrices $A$ and $B$, the former corresponding to the {\em row player} and the latter to the {\em column player}. If $B = A^T$, where $A^T$ is the transpose of $A$, the game is called {\em symmetric}. A {\em mixed strategy} for the row player is a probability vector $P \in \mathbb{R}^n$ and a mixed strategy for the column player is a probability vector $Q \in \mathbb{R}^m$. The {\em payoff} to the row player of $P$ against $Q$ is $P \cdot A Q$ and that to the column player is $P \cdot B Q$. Let us denote the space of probability vectors for the row player by $\mathbb{P}$ and the corresponding space for the column player by $\mathbb{Q}$. A Nash equilibrium of a $2$-player game $(A, B)$ is a pair of mixed strategies $P^*$ and $Q^*$ such that all unilateral deviations from these strategies are not profitable, that is, for all $P \in \mathbb{P}$ and $Q \in \mathbb{Q}$, we simultaneously have that
\begin{align*}
P^* \cdot AQ^* &\geq P \cdot AQ^*\\
P^* \cdot BQ^* &\geq P^* \cdot BQ.
\end{align*}
Observe that if the bimatrix game is symmetric, the second inequality is redundant. Let $(C, C^T)$ be a symmetric bimatrix game. If $C$ is a symmetric matrix (that is, if $C^T = C$), $(C, C^T)$ is called a {\em doubly symmetric game} or a {\em coordination game}. We are going to denote symmetric payoff matrices by $S$. We call a Nash equilibrium strategy, say $P^*$, {\em symmetric,} if $(P^*, P^*)$ is an equilibrium.

\subsection{Evolutionary stability}

We are now ready to give the formal definition of an ESS.  

\begin{definition}
Let $(C, C^T)$ be a symmetric bimatrix game. We say that $X^* \in \mathbb{X}$ is an ESS, if 
\begin{align*}
\exists O \subseteq \mathbb{X} \mbox{ } \forall X \in O / \{ X^* \} : X^* \cdot CX > X \cdot CX.
\end{align*}
Here $\mathbb{X}$ is the space of mixed strategies of $(C, C^T)$ (a simplex in $\mathbb{R}^n$ where $n$ is the number of pure strategies) and $O$ is a neighborhood of $X^*$. 
\end{definition}

We note that an ESS is necessarily an {\em isolated} symmetric Nash equilibrium strategy in the sense that no other symmetric Nash equilibrium strategy exists in a neighborhood of an ESS. We have the following characterization of an ESS in doubly symmetric (coordination) games: Let $f : \mathbb{X} \rightarrow \mathbb{R}$ where $\mathbb{X}$ is a subset of $\mathbb{R}^n$. Recall that $X^* \in \mathbb{X}$ is a strict local maximum of $f$ if there exists a neighborhood $O$ of $X^*$ such that, for all $Y \in O/\{X^*\}, f(X^*) > f(Y)$. Note further that a {\em standard quadratic program} consists of finding maximizers of a quadratic form over the standard simplex.

\begin{theorem}
\label{symmetric_ESS_characterization_1}
$X^*$ is an ESS of $(S, S)$ if and only if $X^*$ is a strict local maximum of the standard quadratic program
\begin{align*}
\mbox{ maximize } &\frac{1}{2} X \cdot S X\\
\mbox{subject to }  & X \in \mathbb{X}
\end{align*}
where $\mathbb{X}$ is the simplex of mixed strategies of $(S, S)$.
\end{theorem}

Theorem \ref{symmetric_ESS_characterization_1} is implicit in the literature as it amounts to a characterization of evolutionary stability based on the well-known notion of {\em strict local efficiency} (cf. \citep[p. 56]{Weibull}).

\section{Results}
\label{results}

We have the following characterization of an ESS. 

\begin{lemma}
\label{ESS_characterization}
$X^*$ is an ESS of $(C, C^T)$ if and only if it is a strict local maximum of
\begin{align}
\frac{1}{2}Y \cdot (C + C^T) Y - X^* \cdot CY, Y \in \mathbb{X}\label{qpqp}
\end{align}
where $\mathbb{X}$ is the simplex in $\mathbb{R}^n$ and $n$ is the size of $C$.
\end{lemma}

\begin{proof}
Straightforward from the definition of an ESS, and the elementary observations that 
\begin{align*}
C = \frac{1}{2} (C + C^T) + \frac{1}{2} (C - C^T)
\end{align*}
and that since $C - C^T$ is skew-symmetric, $Y \cdot (C - C^T) Y = 0$ for all $Y$.
\end{proof}

As a consequence of Lemma \ref{ESS_characterization} we obtain the following immediate fact.

\begin{lemma}
\label{local_maximum_hardness}
The problem of detecting\footnote{In this paper, we use the terms {\em detecting} and {\em recognizing} as synonyms.} a strict local maximum of a standard quadratic program is {\bf coNP}-hard.
\end{lemma}

\begin{proof}
\cite{Etessami} and \cite{Nisan-ESS} show that the problem of detecting an ESS is {\bf coNP}-complete. Since \eqref{qpqp} can be readily brought into the form of a standard quadratic program by a rank two update (cf. \citep{Bomze}), the problem of detecting an ESS reduces (in the sense of a {\em Karp reduction}) to the problem of detecting a strict local maximum of a standard quadratic program in light of Lemma \ref{ESS_characterization}.
\end{proof}

\begin{theorem}
\label{fundamental_complexity}
The problem of detecting an ESS in a doubly symmetric bimatrix game is {\bf coNP}-complete.
\end{theorem}

\begin{proof}
In light of Lemma \ref{local_maximum_hardness} and Theorem \ref{symmetric_ESS_characterization_1}, which together show that the problem of detecting an ESS in a doubly symmetric game is {\bf coNP}-hard, it suffices to show that the problem of detecting an ESS in a doubly symmetric bimatrix game is in {\bf coNP}. But this is a simple implication of the fact that detecting an ESS in a symmetric bimatrix is {\bf coNP}-complete \citep{Etessami, Nisan-ESS} and that the class of doubly symmetric bimatrix games is not but a subset of the class of symmetric bimatrix games. 
\end{proof}

Theorem \ref{fundamental_complexity} implies the following corollary.

\begin{corollary}
\label{rd_corollary}
The problem of detecting an asymptotically stable equilibrium point of the continuous-time replicator dynamic in a doubly symmetric game is {\bf coNP}-complete.
\end{corollary}

\begin{proof}
In doubly symmetric bimatrix games, the notions of asymptotic stability under the replicator dynamic and that of evolutionary stability coincide \citep{Hofbauer-Sigmund}.
\end{proof}

\section*{Acknowledgments}

I would like to thank Bill Sandholm for comments on a previous draft. I would also like to thank Miltos Anagnostou for hosting me at NTUA and for various helpful discussions.

\bibliographystyle{abbrvnat}
\bibliography{real}

\end{document}